\documentclass[onecolumn,journal]{IEEEtran}
\usepackage{graphicx}
\usepackage[T1]{fontenc}
\usepackage{url}
\usepackage[cmex10]{amsmath}
\usepackage{bm}
\usepackage{amssymb,amsthm}
\usepackage{thmtools}
\usepackage{cite}
\usepackage{textcomp}
\usepackage{siunitx}
\usepackage{color}
\usepackage{enumerate}
\usepackage{subfigure}
\usepackage{cases}
\usepackage{bm}
\usepackage{algorithm}
\usepackage{algorithmicx}
\usepackage{algpseudocode}
\usepackage{multirow}
\usepackage{authblk}

\newtheorem{theorem}{Theorem}

\newtheorem{corollary}{Corollary}

\theoremstyle{definition}

\newtheorem{definition}{Definition}

\theoremstyle{remark}
\newtheorem{remark}{Remark}

\IEEEoverridecommandlockouts

\begin{document}

\sloppy
 
\title{Coded MapReduce} 

\author[*]{Songze~Li \thanks{Parts of this work were presented in 53rd Allerton Conference, Sept. 2015 \cite{LMA_all}.}}
\author[$\dag$]{Mohammad~Ali~Maddah-Ali}
\author[*]{A.~Salman~Avestimehr}
\affil[*]{University of Southern California, Los Angeles, CA, USA }
\affil[$\dag$]{Bell Labs, Holmdel, NJ, USA}

\maketitle

\begin{abstract} 
MapReduce is a commonly used framework for executing data-intensive jobs on distributed server clusters. We introduce a variant implementation of MapReduce, namely ``Coded MapReduce'', to substantially reduce the inter-server communication load for the shuffling phase of MapReduce, and thus accelerating its execution. The proposed Coded MapReduce exploits the repetitive mapping of data blocks at different servers to create coding opportunities in the shuffling phase to exchange (key,value) pairs among servers much more efficiently. We demonstrate that Coded MapReduce can cut down the total inter-server communication load by a \emph{multiplicative} factor that grows linearly with the number of servers in the system and it achieves the minimum communication load within a constant multiplicative factor. We also analyze the tradeoff between the ``computation load'' and the ``communication load'' of Coded MapReduce.
\end{abstract}

\section{Introduction}
MapReduce~\cite{dean2004mapreduce} is a programming model that enables the distributed processing of large-scale datasets on a cluster of commodity servers. The desirable features of MapRuduce, such as scalability, simplicity and fault-tolerance~\cite{lee2012parallel,jiang2010performance}, have made this framework popular to perform data-intensive tasks in text/graph processing, machine learning and bioinformatics. Simply speaking, in MapReduce framework, each of the input data blocks, stored across the distributed servers (e.g. HDFS in Hadoop \cite{had}), is assigned to a worker node to process. The processing \emph{maps} the input blocks to some intermediate (key,value) pairs. In the next step, referred as \emph{data shuffling}, the intermediate (key,value) pairs are transferred via inter-server communication links to a set of processors to be \emph{reduced} to the final results. 

The data shuffling phase is one of the key bottlenecks for improving the runtime performance of MapReduce. In fact, as observed in~\cite{chowdhury2011managing}, using a Hadoop~\cite{had} cluster, 33\% of the job execution time is spent on data shuffling. As a result, many optimization methods including combining intermediate (key,value) pairs before shuffling~\cite{dean2004mapreduce,rajaraman2011mining}, optimal flow scheduling across network paths~\cite{greenberg2009vl2,al2010hedera} and employment of distributed cache memories~\cite{zhang2009accelerating,ekanayake2010twister}, have been proposed to accelerate the data shuffling phase of MapReduce programs.

In this paper, we introduce ``Coded MapReduce'', a new framework that \emph{enables} and \emph{exploits} a particular form of coding to significantly reduce the communication load of the shuffling phase. Compared to the conventional MapReduce approach, we demonstrate that Coded MapReduce can substantially cut down the inter-server communication load by a \mbox{\emph{multiplicative}} factor that grows linearly with the number of servers in the system. We also show that Coded MapReduce achieves the minimum shuffling load within a constant multiplicative factor regardless of the system parameters.

Coded MapReduce exploits the repetitive mappings of the same data block at different servers, in order to enable coding. More specifically, for a server cluster with $K$ servers, interconnected through a multicast LAN network, and repetitive mapping of each data block on $r$ fraction of the servers, we propose a particular strategy to assign the Map tasks, such that in the shuffling phase the data demand of each group of approximately $rK>1$ servers can be satisfied by a \emph{single coded transmission}, which we call a \emph{coded multicast opportunity}. This would provide a multiplicative coding gain of $rK$, which scales linearly with the number of servers. Furthermore, even if repetition of data blocks is very small (for example, if each data block is mapped over only one other server, i.e. $rK = 2$), we demonstrate that Coded MapReduce can still substantially reduce the communication load (for example by more than 50\%). As a result, a minor increase in the replication of mapping data blocks at different servers can largely impact the shuffling load of MapReduce.

The idea of Coded MapReduce is inspired by recent results on cache networks, arguing that cache memories can be used not only to deliver part of the contents locally, but also to create some coding opportunities and reduce the traffic significantly~\cite{maddah2014fundamental,maddah2013decentralized}. Surprisingly the gain of coding in reducing the traffic is significantly larger than the gain of local delivery, and indeed it can grow with the size of the network. This result has been further generalized to various cache networks in~\cite{ji2014fundamental,karamchandani2014hierarchical}. Interestingly, we demonstrate that similar coding opportunities can also be created in the MapReduce framework via a careful assignment of repetitive Map tasks across servers. 

The efficiency of Coded MapReduce in data shuffling builds upon the repetitive executions of the Map tasks, which however requires more processing time for the server cluster to compute the intermediate results. We analytically and numerically evaluate the tradeoff between the Map tasks processing time and the inter-server communication load using Coded MapReduce, based on which clients can choose appropriate operating points to minimize the overall execution time.

The rest of the paper is organized as follows. In Section II, we introduce a MapReduce framework and define the inter-server communication load for the shuffling phase. We then motivate Coded MapReduce through a word-counting example in Section III. Lower bound on the minimum shuffling load (proved in Section VI) and the shuffling load achieved by Coded MapReduce are presented in Section IV, where Coded MapReduce are shown to be approximately optimal. We formally describe the Coded MapReduce framework in Section V, and then analyze its impact on reducing the inter-server communication load. We explore the tradeoff between the computation load of Map tasks and the communication load using Coded MapReduce in Section~VII. We finally conclude the paper in Section VIII.

\section{Problem Statement}
For an input file, the job is to evaluate $Q$ pairs of (key,value)'s, denoted by $\{(W_q, u_q)\}_{q=1}^Q$. The keys $W_q$, $q \in \{1, \ldots, Q\}$ are given. The corresponding output values $u_q$, $q \in \{1,\ldots,Q\}$, are evaluated from the input file. The input file consists of $N$ disjoint subfiles (data blocks). The available processing resources include a cluster of $K$ servers.\\ 

\noindent {\bf Example (World Counting).}  Consider a canonical word-counting job of finding the numbers of times 4 given words, denoted by $A$, $B$, $C$ and $D$,  appeared in a book with $N=12$ chapters using $K=4$ servers. Each word itself represents a unique key. Thus we have $Q=4$ keys each corresponding to a word to be counted and $K=4$ servers. The final output results are 4 (key,value) pairs $(A,a)$, $(B,b)$, $(C,c)$ and $(D,d)$ indicating that there are $a$ occurrences of $A$, $b$ occurrences of $B$, $c$ occurrences of $C$ and $d$ occurrences of $D$ in the book. $\hfill \square$

A MapReduce approach~\cite{dean2004mapreduce} carries out the job distributedly across the available $K$ servers. As illustrated in Fig.~\ref{fig:flow}, the MapReduce implementation consists of 3 steps: Map Tasks Assignment, Map Tasks Execution and Reduce Tasks, run over the set of $K$ servers. In what follows we describe each of these 3 steps in detail. 

\begin{figure}[htbp]
   \centering
   \includegraphics[width=0.6\textwidth]{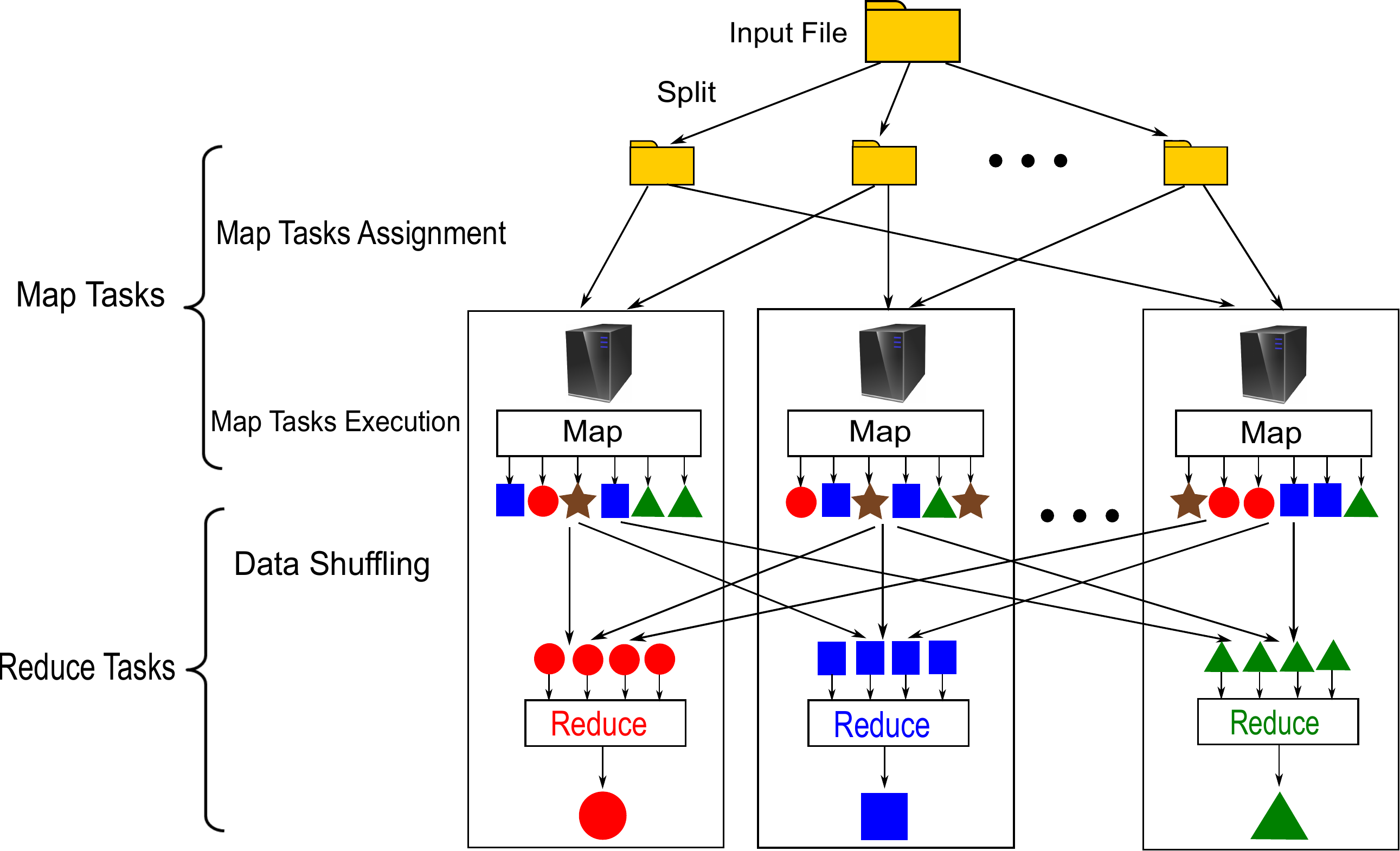}
   \caption{The workflow of a generic MapReduce implementation on $K$ servers. The Map tasks consist of Map tasks assignment and Map tasks execution, and the Reduce tasks consist of data shuffling and the executions of the Reduce functions. Notice that each server executes both Map and Reduce tasks.}
   \label{fig:flow}
\end{figure}

\subsection*{Step 1: Map Tasks Assignment}    
To start a MapReduce execution, a master controller (JobTracker) assigns the tasks of mapping the $N$ subfiles among the available $K$ servers. Each server will focus on processing its assigned subfiles (see details in Step 2). For a given design parameter $p \in \left\{\frac{1}{K},\frac{2}{K},\ldots,1\right\}$, the assignment is done such that each subfile is assigned to be mapped at $pK$ servers and each server is assigned $pN$ subfiles (assuming that $pN$ is an integer). While conventional MapReduce implementations require each subfile to be mapped at exactly one server ($p = \frac{1}{K}$), the master controller can repetitively schedule two or more servers to map the same subfile. 

We denote the set of indices of the subfiles assigned to Server $k$ as $\mathcal{M}_k$ and the set of indices of the servers Subfile $n$ is assigned to as $\mathcal{A}_n$. Either $\{\mathcal{M}_k\}_{k=1}^K$ or $\{\mathcal{A}_n\}_{n=1}^N$ completely specifies the Map tasks assignment.

\noindent {\bf Example (Word-Counting: Map Tasks Assignment).} In our word-counting example, $Q=K=4$, $N=12$. For instance $p$ can be set to $\frac{2}{K}=\frac{1}{2}$, i.e., each subfile is assigned to be mapped at $pK=2$ servers.  We can employ a \emph{naive} Map tasks assignment that assigns each of the first $pN$ (6 in this example) chapters to each of the first $pK$ servers, each of the second $pN$ chapters to each of the second $pK$ servers and so forth. For such an assignment,  ${\mathcal{M}_1=\mathcal{M}_2=\{1,2,3,4,5,6\}}$, ${\mathcal{M}_3=\mathcal{M}_4=\{7,8,9,10,11,12\}}$. That is, Server 1 and Server 2 are assigned to map Chapters 1, 2, 3, 4, 5, 6 and Server 3 and Server 4 are assigned to map Chapters 7, 8, 9, 10, 11, 12. $\hfill \square$

\subsection*{Step 2: Map Tasks Execution}
After the Map tasks assignment, Server $k$, $k \in \{1,\ldots,K\}$, starts to process each Subfile $n$ for all $n \in \mathcal{M}_k$, separately, maps it to a set of $Q$ intermediate (key,value) pairs $(W_1,v_{1n})[n],\ldots,(W_Q,v_{Qn})[n]$. For each $q \in \{1,\ldots,Q\}$ and the intermediate pair $(W_q,v_{qn})[n]$, $W_q$ is the key and ${v_{qn} \in \mathbb{F}_{2^F}}$ is the intermediate value of $W_q$ in Subfile $n$. 

Conventionally, the Map tasks execution continues until that each of the $N$ subfiles has been mapped to (key,value) pairs at one server. Here we introduce a key design parameter $r$, $r \in \left\{\frac{1}{K},\frac{2}{K},\ldots, p\right\}$ and enforce that for all $n \in \{1,\ldots,N\}$, as soon as $rK$ servers finish mapping Subfile $n$, the rest of the servers in $\mathcal{A}_n$ abort the task of mapping Subfile $n$. Thus each subfile is mapped at $rK$ servers by the end of Map tasks execution. We denote the set of indices of the servers that finish mapping Subfile $n$ as $\mathcal{A}_n' \subseteq \mathcal{A}_n$, where $|\mathcal{A}_n'|= rK$ for all $n$. $r$ reflects the redundancy of mapping the same subfile across different servers, and large $r$ can cause additional delay to the overall Map tasks execution. However, as we demonstrate later, having repetitive mapping outcomes at distinct servers can help to significantly reduce the communication load of the shuffling phase between Map and Reduce tasks.\\

\begin{figure}[htbp]
   \centering
   \includegraphics[width=0.98\textwidth]{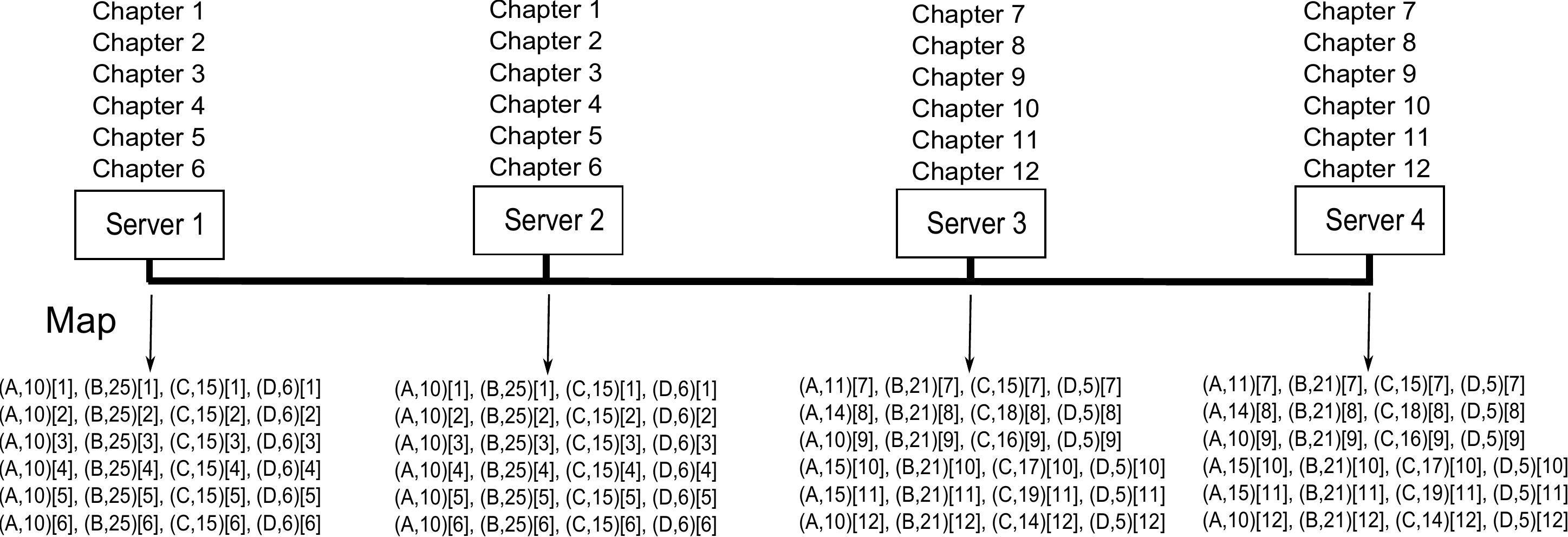}
   \caption{The Map tasks execution for the naive Map tasks assignment. Each server counts the number of occurrences of $A$, $B$, $C$ and $D$ in the assigned chapters, and generates 4 (key,value) pairs. For example, a pair $(A,10)[1]$ indicates that $A$ appears 10 times in Chapter 1.}
   \label{fig:exe1}
\end{figure}

\noindent {\bf Example (Word-Counting: Map Tasks Execution).} To continue our word-counting example, having adopted the naive Map tasks assignment, the master controller sets $r=p=\frac{1}{2}$ such that each server has to finish mapping all of the assigned $pN = 6$ chapters. As illustrated in Fig.~\ref{fig:exe1}, Server $k$, $k \in\{1,2,3,4\}$, maps Chapter $n$ in $\mathcal{M}_k$ into 4 intermediate (key,value) pairs $(A,a_n)[n]$, $(B,b_n)[n]$, $(C,c_n)[n]$ and $(D,d_n)[n]$. For the pair $(A,a_n)[n]$, $A$ is the key, $a_n$ is the value such that $a_n$ is the number of occurrences of $A$ in Chapter $n$ \footnote{While a conventional implementation of the Map function in word-counting jobs maps each encountered word say $A$ into an intermediate pair $(A,1)$, we employ combiners~\cite{dean2004mapreduce} to generate the total number of occurrences of interested words in a chapter.}. The definitions of pairs with keys $B$, $C$ and $D$ follow similarly.  $\hfill \square$

\subsection*{Step 3: Reduce Tasks}
The goal of this step is to evaluate the final output (key, value) pairs, in a distributed fashion,  from the intermediate (key,value) pairs, evaluated in Step 2, i.e. Map tasks execution.
One reducer is responsible for evaluating one of the $Q$ keys, namely $W_1,\ldots,W_Q$, and thus a total of $Q$ reducers are required to be executed on $K$ servers. The executions of these reducers are distributed uniformly across $K$ servers. For simplicity, we assume that $\frac{Q}{K}$ is an integer and each server executes a \emph{disjoint} set of $\frac{Q}{K}$ reducers. We denote the reducer distribution as $\mathcal{D}\triangleq \left(\mathcal{W}_1,\ldots,\mathcal{W}_K\right)$, where $\mathcal{W}_k$, $k \in \{1,\ldots,K\}$, is the set of the indices of the keys evaluated at Server $k$. Thus any valid reducer distribution must satisfy:
\begin{enumerate}
\item $|\mathcal{W}_k|=\frac{Q}{K}$ for all $k \in \{1,\ldots,K\}$,
\item $\underset{k = 1,\ldots,K}{\cup}\mathcal{W}_k = \{1,\ldots,Q\}$,
\item $\mathcal{W}_k \cap \mathcal{W}_{k'} =\emptyset$ for $k \neq k'$.
\end{enumerate}

To evaluate the key $W_q$ for some $q \in \mathcal{W}_k$, the corresponding reducer at Server $k$ needs the values of $W_q$ in all $N$ subfiles $\{v_{qn}\}_{n=1}^N$. By the end of the Map tasks execution, Server $k$ has mapped a subset $\mathcal{M}_k'$ of its assigned subfiles $\mathcal{M}_k$ and knows the values ${\{v_{qn}:n \in \mathcal{M}_k'\}}$, thus the remaining values $\{v_{qn}: n \notin \mathcal{M}_k'\}$ are needed from other servers to execute the reducer for $W_q$. We consider a multicast LAN network \footnote{While current main-stream cloud computing platforms like Amazon EC2 do not support L2 multicast/broadcast, we will demonstrate the significant benefit of enabling multicast in the efficiency of data movement using our proposed Coded MapReduce scheme.} where the $K$ servers are interconnected through a shared link. Next, we formally define the data shuffling scheme to exchange the required values for reduction and the resulting communication load.

\begin{definition}[Data Shuffling]
A data shuffling scheme for a MapReduce job with $N$ subfiles, $Q$ keys, $K$ servers and parameters $p,r, \{\mathcal{M}_k'\}_{k=1}^K$, is defined as follows:
\begin{itemize}
\item At each time slot $t \in \mathbb{N}$, one of the $K$ servers, say Server $k$, creates a message of $F$ bits as a function of the values from the subfiles that are mapped by that server (i.e., $\mathcal{M}_k'$), denoted by ${X_{k,t} = \phi_{k,t}\left(\{v_{qn}:q \in \{1,\ldots,Q\},n \in \mathcal{M}_k'\}\right)}$, and sends it to all other servers via the multicast LAN network interconnecting them.
\item The communication process continues for $T$ times slots, until Server $k$, for all $k \in \{1,\ldots,K\}$, is able to successfully construct the required values to execute the reducers for the keys in $\mathcal{W}_k$, based on the messages it receives from other servers and its own Map outcomes for the keys in $\mathcal{W}_k$ (i.e., $\{v_{qn}:q \in \mathcal{W}_k, n \in \mathcal{M}_k'\}$). 
\item The communication load of a data shuffling scheme, denoted by $L$, is defined as
 $L \triangleq \mathbb{E}\{T\}$, i.e., the \emph{average} number of  inter-server communication time slots required for that scheme, where the average is taken over all possible Map task executions of each subfile (i.e., which of the $rK$ servers out of the $pK$ servers have mapped each subfile).  $\hfill\square$
\end{itemize}
\end{definition}

\begin{remark}
Given an instance of the Map tasks execution $\{\mathcal{M}_k'\}_{k=1}^K$, each server knows the intermediate results in mapped subfiles for all keys. Therefore, for all valid reducer distributions (i.e., $K$ mutually disjoint subsets each with $\frac{Q}{K}$ keys are reduced respectively at the $K$ servers), their communication loads are identical. In other words, for a particular Map tasks execution outcome, the communication load is independent of the reducer distribution.    $\hfill \square$
\end{remark}

Given a MapReduce job of using a cluster of $K$ servers to evaluate $Q$ keys in an input file consisting of $N$ subfiles where each subfile is assigned to be mapped at $pK$ servers for some $p \in \{\frac{1}{K},\ldots,1\}$, we are interested in the problem of minimizing the communication load over the Map tasks assignment $\mathcal{M}_1,\ldots,\mathcal{M}_K$ and the data shuffling scheme. We define $L^*(r)$ as the minimum communication load when each subfile is mapped at $rK$ servers out of the $pK$ servers it is assigned to.

As a baseline, we can consider the conventional MapReduce approach where each subfile is assigned to and mapped at only one server ($pK=rK=1$). In this setting, each server maps $\frac{N}{K}$ subfiles, obtaining $\frac{N}{K}$ intermediate values for each of its assigned keys. Since each server reduces $\frac{Q}{K}$ keys, the communication load for the conventional approach is
\begin{equation}\label{eq:conv}
L_{\text{conv}}= K \cdot \frac{Q}{K} \cdot \left(N-\frac{N}{K}\right) = QN\left(1-\frac{1}{K}\right).
\end{equation}

By increasing $p$ and $r$ beyond $\frac{1}{K}$, each subfile is now repeatitively mapped at $rK>1$ servers and the total number of executed Map tasks increases by $rK$ times: $\sum \limits_{k=1}^K |\mathcal{M}'_k| = rKN$. Server $k$, $k \in \{1,\ldots,K\}$, needs another $\frac{Q}{K}\left(N-|\mathcal{M}_k'|\right)$ intermediate values to execute its $\frac{Q}{K}$ reducers. These data requests can for example be satisfied by a simple \emph{uncoded} data shuffling scheme such that each of the required intermediate values is sent over the shared link at a time, achieving the following communication load 
\begin{equation}\label{eq:uncode}
L_{\text{uncoded}}(r) = \frac{Q}{K}\sum \limits_{k=1}^{K}\left(N-\mathbb{E}\left\{\left|\mathcal{M}_k'\right|\right\}\right)= QN(1-r). 
\end{equation}

Let us illustrate the uncoded scheme by describing the shuffling phase of our running word-counting example.

\noindent {\bf Example (Word-Counting: Data Shuffling via Uncoded Scheme).} Since ${Q=K=4}$, each server executes one reducer, i.e., Server 1 evaluates $A$, Server 2 evaluates $B$, Server 3 evaluates $C$, and Server 4 evaluates $D$. 

Based on the results of the Map tasks execution (see Fig.~\ref{fig:exe1}), Server 1 and Server 2 need the values of $A$ and $B$ respectively in Chapters 7, 8, 9, 10, 11, 12. Server 3 and Server 4 need the values of $C$ and $D$ respectively in Chapters 1, 2, 3, 4, 5, 6. An uncoded data shuffling is carried out as follows:
\begin{enumerate}
\item Server 3 sends pairs $(A,11)[7]$, $(A,14)[8]$, $(A,10)[9]$, $(A,15)[10]$, $(A,15)[11]$, $(A,10)[12]$ and then $(B,21)[7]$, $(B,21)[8]$, $(B,21)[9]$, $(B,21)[10]$, $(B,21)[11]$, $(B,21)[12]$.
\item Server 2 sends pairs $(C,15)[1]$, $(C,15)[2]$, $(C,15)[3]$, $(C,15)[4]$, $(C,15)[5]$, $(C,15)[6]$ and then $(D,6)[1]$, $(D,6)[2]$, $(D,6)[3]$, $(D,6)[4]$, $(D,6)[5]$, $(D,6)[6]$.
\end{enumerate}

After the communication every server knows the values of its interested word in all 12 chapters, and these values are passed into a reducer to compute the final result. The shuffling phase lasts for 24 time slots and thus the communication load of the uncoded scheme is 24, which is consistent with equation~(\ref{eq:uncode}) for $N=12$, $Q=4$ and $r = \frac{1}{2}$. Notice that, if we had employed the conventional MapReduce approach where each subfile is mapped at only one server, then the communication load of this particular job would have been 36 (this is obtained by setting  $N=12$ and $Q=K=4$ in equation~(\ref{eq:conv})). $\hfill \square$

Comparing equations~(\ref{eq:conv}) and (\ref{eq:uncode}), we notice that by repeatedly mapping the same subfile at more than one server ($rK \geq 2$), the communication load of the shuffling phase in MapReduce can be improved by a factor of $\frac{1-\frac{1}{K}}{1-r}$, when using a simple uncoded scheme. This improvement in the communication load results from the fact that by mapping each subfile repeatedly at multiple servers, the servers know values from $rK$ times more subfiles than the conventional approach, thus requiring less values communicated during data shuffling to execute their reducers. We denote this gain as the \emph{repetition gain}, which is due to knowing values from more subfiles \emph{locally} at each server. 

As we will show next, in addition to the repetition gain, repeatedly mapping each subfile at multiple servers can have a much more significant impact on reducing the communication load of the data shuffling, which can be achieved by \emph{a more careful assignment of Map tasks to servers} and \emph{exploiting coding in the data shuffling phase}. We will next illustrate this through a motivating example, which forms the basis of the general Coded MapReduce framework that we will later present in Section~V.

\section{Coded MapReduce: A Motivating Example}
In this section we motivate Coded MapReduce via a simple example. In particular, we demonstrate through this example that, by carefully assigning Map tasks to the servers, there will be novel coding opportunities in the shuffling phase that can be utilized to significantly reduce the inter-server communication load of MapReduce.

We consider the same word-counting job of counting 4 words $A$, $B$, $C$ and $D$ in a book with 12 chapters using 4 servers. While maintaining the same number subfiles to map at each server (6 in this case), we consider a new Map tasks assignment as follows. \\

\noindent \emph{Map Tasks Assignment}

Instead of using the naive assignment, the master controller assigns the Map tasks as follows: $\mathcal{M}_1=\{1,2,3,4,5,6\}$, $\mathcal{M}_2=\{1,2,7,8,9,10\}$, $\mathcal{M}_3=\{3,4,7,8,11,12\}$, $\mathcal{M}_4=\{5,6,9,10,11,12\}$. Notice that in this assignment each chapter is assigned to exactly two servers and every two servers share exactly two chapters.\\

\begin{figure}[htbp]
   \centering
   \includegraphics[width=0.98\textwidth]{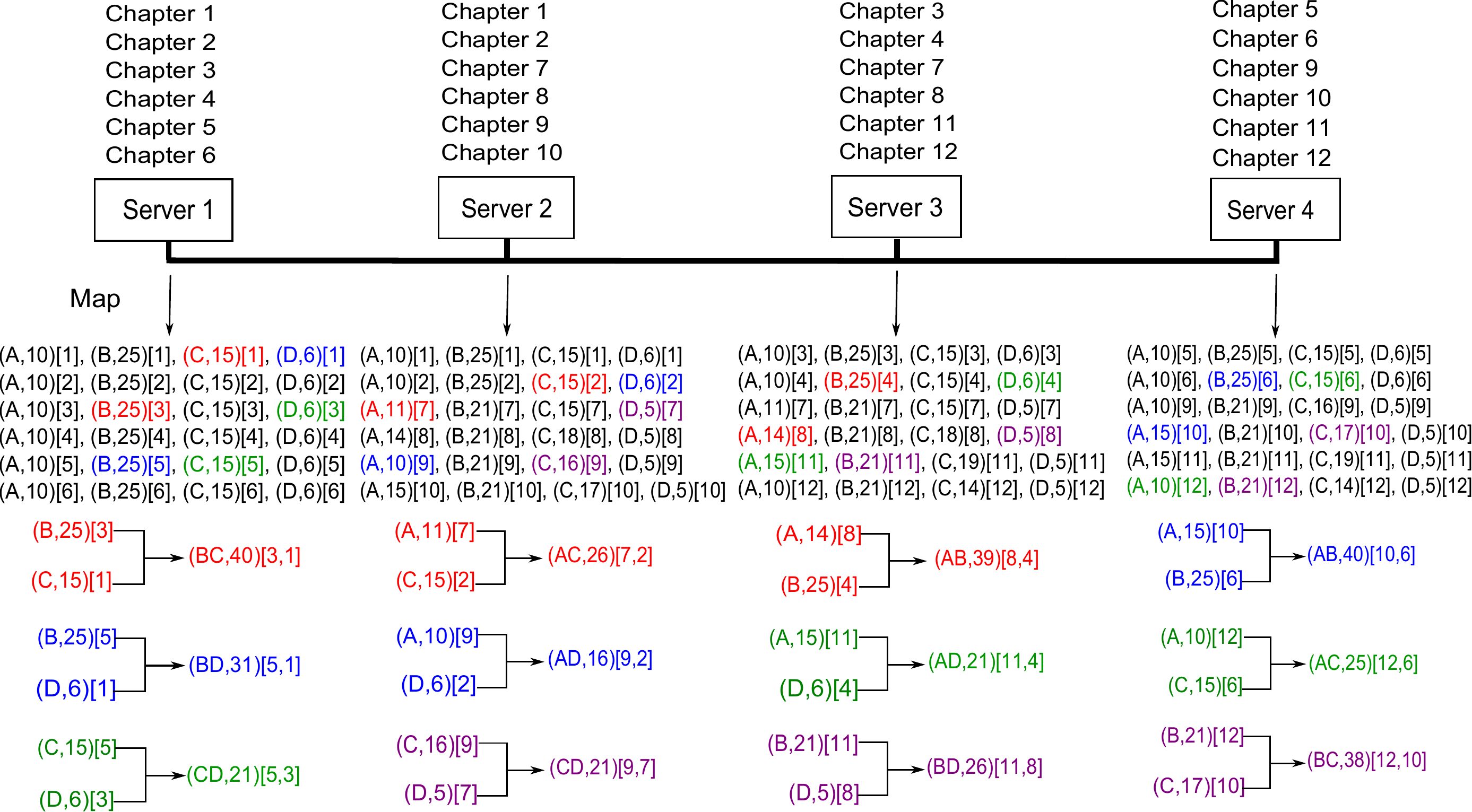}
   \caption{The Map tasks execution of Coded MapReduce for a word-counting job. Having generated 24 (key,value) pairs, one for each word and each of the assigned chapters, each server further generates 3 coded (key,value) pairs, each with a different color by summing up the last value of two intermediate pairs with the same color. For example at Server 1, the red coded pair $(BC,40)[3,1]$ generated from $(B,25)[3]$ and $(C,15)[1]$ indicates that the number of occurrences of $B$ in Chapter 3 and the number of occurrences of $C$ in Chapter 1 sum up to 40.}
   \label{fig:exe2}
\end{figure}

\noindent \emph{Map Tasks Execution}

The master controller sets $r=p$ such that each server has to finish mapping all assigned chapters. The execution of the Map tasks is different from that of the naive assignment such that after generating 4 intermediate (key,value) pairs for each of the assigned chapters, each server generates 3 additional coded (key,value) pairs as follows (see Fig.~\ref{fig:exe2}):
\begin{itemize}
\item Server 1 adds up the values of $(B,25)[3]$ and $(C,15)[1]$ to generate a pair $(BC,40)[3,1]$, adds up the values of $(B,25)[5]$ and $(D,6)[1]$ to generate another pair $(BD,31)[5,1]$, and adds up the values of $(C,15)[5]$ and $(D,6)[3]$ to generate a third pair $(CD,21)[5,3]$,
\item Server 2 adds up the values of $(A,11)[7]$ and $(C,15)[2]$ to generate a pair $(AC,26)[7,2]$, adds up the values of $(A,10)[9]$ and $(D,6)[2]$ to generate another pair $(AD,16)[9,2]$, and adds up the values of $(C,16)[9]$ and $(D,5)[7]$ to generate a third pair $(CD,21)[9,7]$,
\item Server 3 adds up the values of $(A,14)[8]$ and $(B,25)[4]$ to generate a pair $(AB,39)[8,4]$, adds up the values of $(A,15)[11]$ and $(D,6)[4]$ to generate another pair $(AD,21)[11,4]$, and adds up the values of $(B,21)[11]$ and $(D,5)[8]$ to generate a third pair $(BD,26)[11,8]$,
\item Server 4 adds up the values of $(A,15)[10]$ and $(B,25)[6]$ to generate a pair $(AB,40)[10,6]$, adds up the values of $(A,10)[12]$ and $(C,15)[6]$ to generate another pair $(AC,25)[12,6]$, and adds up the values of $(B,21)[12]$ and $(C,17)[10]$ to generate a third pair $(BC,38)[12,10]$.
\end{itemize} 

A coded pair $(W_1W_2,x)[n_1,n_2]$ has key $W_1W_2$ and value $x$, and it indicates that there are $x$ occurrences in total of Word $W_1$ in Chapter $n_1$ and Word $W_2$ in Chapter $n_2$. \\

The Reduce tasks are distributed the same as before: Server 1 evaluates $A$, Server 2 evaluates $B$, Server 3 evaluates $C$ and Server 4 evaluates $D$. \\

\noindent \emph{Data Shuffling}

After executing the Map tasks, values from 6 chapters are missing at each server to execute the reducer. To fulfill the data requests for reduction, the data shuffling is carried out such that each server sends the 3 coded pairs generated during Map tasks execution:
\begin{enumerate}
\item Server 1 sends pairs $(BC,40)[3,1]$, $(BD,31)[5,1]$ and $(CD,21)[5,3]$,
\item Server 2 sends pairs $(AC,26)[7,2]$, $(AD,16)[9,2]$ and $(CD,21)[9,7]$,
\item Server 3 sends pairs $(AB,39)[8,4]$, $(AD,21)[11,4]$ and $(BD,26)[11,8]$,
\item Server 4 sends pairs $(AB,40)[10,6]$, $(AC,25)[12,6]$ and $(BC,38)[12,10]$.
\end{enumerate} 

Having received all coded pairs, each server performs an additional decoding operation before executing the final Reduce function:
\begin{enumerate}
\item Server 1 subtracts the values of $(C,15)[2]$, $(D,6)[2]$, $(B,25)[4]$, $(D,6)[4]$, $(B,25)[6]$ and $(C,15)[6]$ from the values of $(AC,26)[7,2]$, $(AD,16)[9,2]$, $(AB,39)[8,4]$, $(AD,21)[11,4]$, $(AB,40)[10,6]$ and $(AC,25)[12,6]$ respectively to decode 6 pairs it needs to count $A$: $(A,11)[7]$, $(A,10)[9]$, $(A,14)[8]$, $(A,15)[11]$, $(A,15)[10]$ and $(A,10)[12]$,
\item Server 2 subtracts the values of $(C,15)[1]$, $(D,6)[1]$, $(A,14)[8]$, $(D,5)[8]$, $(A,15)[10]$ and $(C,17)[10]$ from the values of $(BC,40)[3,1]$, $(BD,31)[5,1]$, $(AB,39)[8,4]$, $(BD,26)[11,8]$, $(AB,40)[10,6]$ and $(BC,38)[12,10]$ respectively to decode 6 pairs it needs to count $B$: $(B,25)[3]$, $(B,25)[5]$, $(B,25)[4]$, $(B,21)[11]$, $(B,25)[6]$ and $(B,21)[12]$,
\item Server 3 subtracts the values of $(B,25)[3]$, $(D,6)[3]$, $(A,11)[7]$, $(D,5)[7]$, $(A,10)[12]$ and $(B,21)[12]$ from the values of $(BC,40)[3,1]$, $(CD,21)[5,3]$, $(AC,26)[7,2]$, $(CD,21)[9,7]$, $(AC,25)[12,6]$ and $(BC,38)[12,10]$ respectively to decode 6 pairs it needs to count $C$: $(C,15)[1]$, $(C,15)[5]$, $(C,15)[2]$, $(C,16)[9]$, $(C,15)[6]$ and $(C,17)[10]$,
\item Server 4 subtracts the values of $(B,25)[5]$, $(C,15)[5]$, $(A,10)[9]$, $(C,16)[9]$, $(A,15)[11]$ and $(B,21)[11]$ from the values of $(BD,31)[5,1]$, $(CD,21)[5,3]$, $(AD,16)[9,2]$, $(CD,21)[9,7]$, $(AD,21)[11,4]$ and $(BD,26)[11,8]$ respectively to decode 6 pairs it needs to count $D$: $(D,6)[1]$, $(D,6)[3]$, $(D,6)[2]$, $(D,5)[7]$, $(D,6)[4]$ and $(D,5)[8]$.
\end{enumerate} 

Now each server knows the value of the interested word in each of the 12 chapters, and passes these values into a Reduce function to generate the final result. Each server accesses the shared link 3 times during the shuffling phase and the total communication load is $12$.

We notice that having successfully communicated the required values for the Reduce functions, the communication load of the proposed Coded MapReduce is 66\% less than that of the conventional MapReduce approach, and 50\% less than that of the naive assignment and uncoded shuffling scheme. The savings in the communication load result from the fact that each of the coded pairs $(W_1W_2,v_{1n_1}+v_{2n_2})[n_1,n_2]$ simultaneously delivers $v_{1n_1}$ to the server counting $W_1$ and $v_{2n_2}$ to the server counting $W_2$, given that $v_{1n_1}$ is known at the server counting $W_2$ and $v_{2n_2}$ is known at the server counting $W_1$ prior to data shuffling. For example because $b_3$ is needed by Server 2 and known at Server 3 and $c_1$ is needed by Server 3 and known at Server 2, having Server 1 send $(BC,b_3+c_1)[3,1]$ simultaneously delivers $b_3$ to Server 2 and $c_1$ to Server 3. We call such opportunities of effectively communicating multiple values through a single use of the shared link as \mbox{\emph{Coded Multicast}}.  Similar ideas of optimally creating and exploiting the coding opportunities have been first used in caching problems~\cite{maddah2014fundamental,maddah2013decentralized,ji2014fundamental}. This type of  coding  is  closely related to the network coding~\cite{ahlswede2000network} problem and the index coding problem~\cite{birk2006coding,bar2011index}. More detailed discussion about the connections between these problems can be found in~\cite{maddah2014fundamental}.

\begin{remark}
For this particular word-counting example, one can save communication load by adding the word counts in different chapters right after Map tasks execution and sending the sum instead of the counts in individual chapters. However, we emphasize that for a general MapReduce job where such pre-reduction might not be applicable (Reduce function is not associative or commutative), the proposed Coded MapReduce scheme, which is presented in detail in Section~V, can still provide a significant gain in reducing the shuffling load. With sufficiently large number of subfiles, this gain scales linearly with the number of servers in the system. $\hfill \square$
\end{remark}

\section{Main Results}
In this section we present and discuss upper and lower bounds on the minimum communication load $L^*(r)$ of a general MapReduce job. We also analytically and numerically compare the communication loads of different Map tasks assignment and data shuffling schemes to demonstrate the substantial reduction in communication load of our proposed Coded MapReduce scheme compared with the conventional MapReduce approach and the uncoded shuffling scheme. 

\begin{theorem}
Consider a job of using $K$ servers to evaluate $Q$ keys in an input file consisting of $N$ subfiles, where each subfile is repetitively assigned to $pK$ servers and is randomly and uniformly mapped at $rK$ of those servers for some $r \in \{\frac{1}{K},\ldots,p\}$. The minimum communication load $L^*(r)$ is bounded as
\begin{equation}\label{eq:miniLoad}
\max\left\{QN\dfrac{1-r}{K-1}, \underset{s \in \{1,\ldots, K\}}{\max}  sQN \left(\dfrac{1}{K}-\dfrac{r}{\left\lfloor \frac{K}{s} \right \rfloor}\right) \right\} \leq  L^*(r) \leq L_{\textup{CMR}}(r)\triangleq \dfrac{QN}{K}\left(\dfrac{1}{r}-1\right)+o(N)
\end{equation}
\end{theorem}

The upper bound $L_{\textup{CMR}}(r)$ of the minimum communication load is achieved by our proposed Coded MapReduce scheme that is described and analyzed in detail in the next section. The two lower bounds on the left hand side of (\ref{eq:miniLoad}) are derived by applying the cut-set bounds on the compound extension of the MapReduce system with multiple valid reducer distributions (i.e., which keys are reduced at which servers), and the detailed proofs are provided in Section VI. 

Next, we demonstrate through the following corollary of Theorem~1 that the proposed \mbox{Coded MapReduce} significantly reduce the inter-server communication load compared to the conventional MapReduce approach.  
\begin{corollary}
Consider a job of using $K$ servers to evaluate $Q$ keys in an input file consisting of $N$ subfiles, where each subfile is repetitively assigned to $pK$ servers according to the Coded MapReduce scheme, which will be described in \textup{Section~V-A}, and is randomly and uniformly mapped at $rK$ of those servers for some $r \in \{\frac{1}{K},\ldots,p\}$. Then, the communication load of the proposed Coded MapReduce scheme $L_{\textup{CMR}}(r)$ satisfies
\begin{equation}\label{eq:compare}
\lim_{N \rightarrow \infty} \frac{L_{\textup{CMR}}(r)}{L_{\textup{conv}}}= \frac{1-r}{1-\frac{1}{K}} \cdot \frac{1}{rK},
\end{equation}
where $L_{\textup{conv}}$ is the communication load of the conventional MapReduce approach, as defined in (\ref{eq:conv}).
\end{corollary}

\begin{figure}[htbp]
   \centering
   \includegraphics[width=0.48\textwidth]{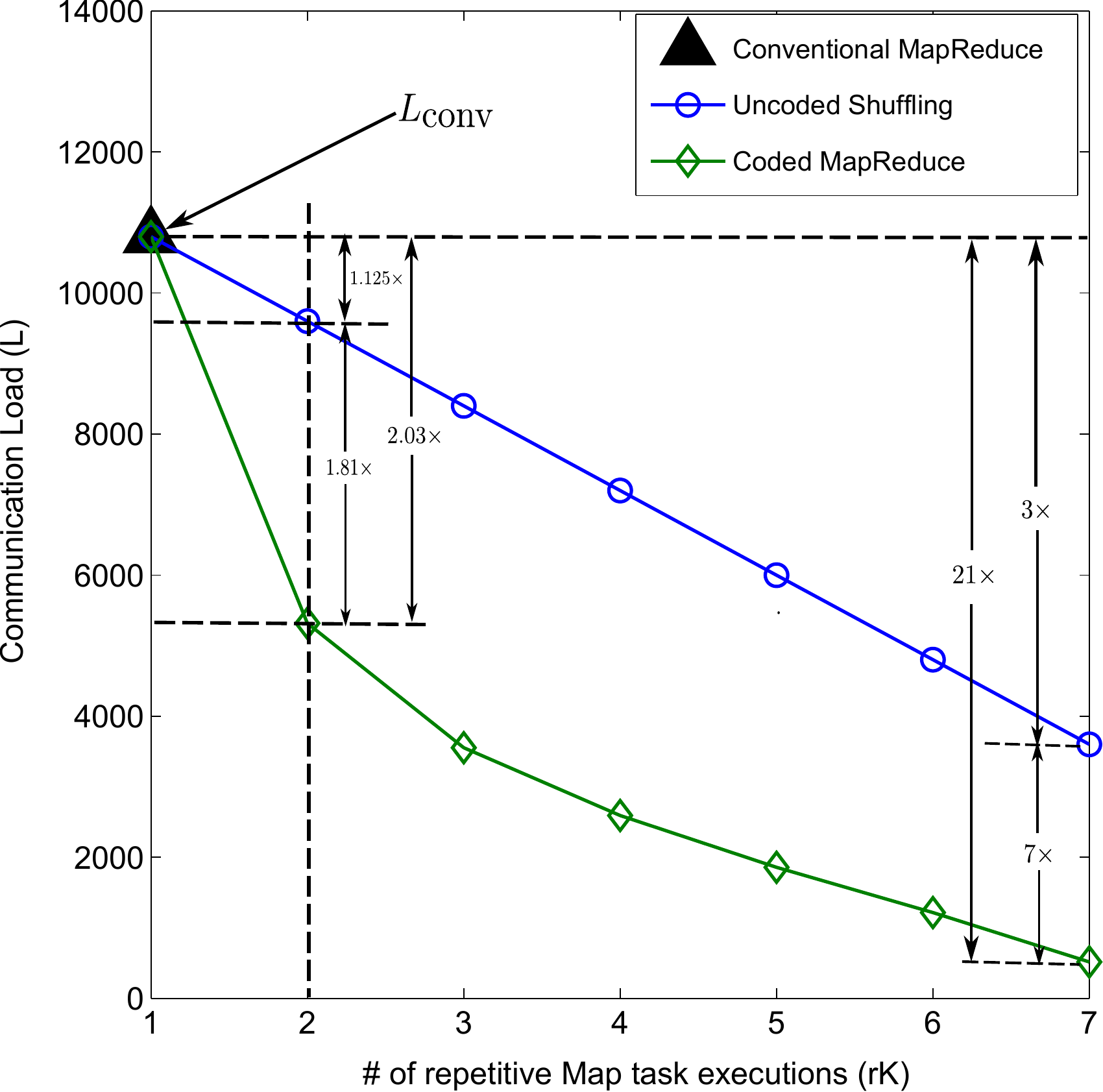}
   \caption{Numerical illustration of the gain from Coded MapReduce in the non-asymptotic region, compared with the conventional MapReduce approach and the uncoded shuffling scheme in reducing the communication load. The performance of the Coded MapReduce is simulated and plotted for a job with $N=1200$ subfiles, $Q=10$ keys, $K=10$ servers and  $pK=7$ repetitive assignments of each subfile.}
   \label{fig:upper}
\end{figure}

\begin{remark}
Corollary~1 demonstrates that compared with the conventional MapReduce approach, for sufficiently large $N$, Coded MapReduce can cut down the communication load
substantially by a multiplicative factor of $\frac{1-\frac{1}{K}}{1-r}(rK) \geq rK$, which grows linearly with the number of servers in the system.  $\hfill \square$
\end{remark} 

\begin{remark}
The first term on the right hand side of (\ref{eq:compare}) (i.e., $\frac{1-r}{1-\frac{1}{K}}$) can be viewed as the \emph{repetition gain}, which is due to knowing values from more subfiles locally at each server, and the second term (i.e., $\frac{1}{rK}$) can be viewed as the coding gain. Note that $\frac{1-r}{1-\frac{1}{K}} \geq 1-r$, hence the repetition gain does not scale with $K$, so the overall gain of Coded MapReduce is mostly due to coding. $\hfill \square$
\end{remark}

\begin{remark}
As an example, we have numerically compared the communication loads required by the conventional \mbox{MapReduce} approach, the uncoded shuffling scheme and the Coded MapReduce for a specific set of parameters in Fig.~\ref{fig:upper}. One can note that for an input file consisting of $1200$ subfiles ($N=1200$), when each subfile is repeatedly assigned to 7 server and actually mapped at 2 servers ($rK=2$), we observe a repetition gain of $1.125 \times$, a coding gain of $1.81 \times$ by Coded MapReduce over the uncoded scheme, and an overall $2.03 \times$ reduction in communication load from the conventional approach by Coded MapReduce. When $rK$ is increased to $7$, i.e., every server has to finish mapping all its assigned subfiles, the repetition gain increases to $3 \times$, the coding gain over the uncoded scheme increases to $7 \times$ and Coded MapReduce achieves an overall $21 \times $ reduction of the communication load compared with the conventional MapReduce approach. Lastly, we notice that the coding gain for this particular job is proportional to $rK$, which is consistent with (\ref{eq:compare}).  $\hfill \square$
\end{remark}

Next we demonstrate via the following theorem that the proposed Coded MapReduce scheme is approximately optimal in the sense that it achieved the minimum communication load within a constant multiplicative factor regardless of the system parameters.
\begin{theorem}
The proposed Coded MapReduce scheme achieves the minimum communication load up to a constant multiplicative factor for any MapReduce job. More precisely:
\begin{equation}
\lim_{N \rightarrow \infty}\frac{L_{\textup{CMR}}(r)}{L^*(r)} < 3+\sqrt{5},
\end{equation}
where $L_{\textup{CMR}}(r)$ is defined in (\ref{eq:miniLoad}). 
\end{theorem}

\begin{proof}
By Theorem~1 we know that 
\begin{align}
\frac{L_{\textup{CMR}}(r)}{L^*(r)}  & \leq \frac{\frac{QN}{K}\left(\frac{1}{r}-1\right) + o(N)}{\underset{s \in \{1,\ldots, K\}}{\max}  sQN \left(\frac{1}{K}-\frac{r}{\left\lfloor \frac{K}{s} \right \rfloor}\right)}\\
&=  \frac{\frac{1}{r}-1 + \frac{K}{Q} \frac{o(N)}{N}}{\underset{s \in \{1,\ldots, K\}}{\max}  s \left(1-\frac{rK}{\left\lfloor \frac{K}{s} \right \rfloor}\right)}. \label{eq:constGap4}
\end{align}

Taking the limit of (\ref{eq:constGap4}) as $N$ goes to infinity, we have
\begin{align}
\lim_{N \rightarrow \infty}\frac{L_{\textup{CMR}}(r)}{L^*(r)} \leq  \frac{\frac{1}{r}-1}{\underset{s \in \{1,\ldots, K\}}{\max}  s \left(1-\frac{rK}{\left\lfloor \frac{K}{s} \right \rfloor}\right)}. \label{eq:constGap5}
\end{align}

We recall that $r \in \{\frac{1}{K},\ldots,1\}$ and proceed to bound (\ref{eq:constGap5}) in the following two regions:
\begin{itemize}
\item $1 \leq \frac{1}{r} < 3+\sqrt{5}$,
\item $\frac{1}{r} > 3+\sqrt{5}$.
\end{itemize}

\noindent 1) $1 \leq \frac{1}{r} < 3+\sqrt{5}$

In this case, we set $s = 1$ in the lower bound of $L^*(r)$ and (\ref{eq:constGap5}) becomes 
\begin{align}
\lim_{N \rightarrow \infty}\frac{L_{\textup{CMR}}(r)}{L^*(r)} &\leq  \frac{\frac{1}{r}-1}{1-r}\\
= \frac{1}{r} < 3+\sqrt{5}. \label{eq:constGap6}
\end{align}

\noindent 2) $\frac{1}{r} > 3+\sqrt{5}$

In this case, we set $s = \left\lfloor \frac{1}{2r}\right\rfloor$ in the lower bound of $L^*(r)$ to obtain
\begin{align}
\lim_{N \rightarrow \infty}\frac{L_{\textup{CMR}}(r)}{L^*(r)} &\leq  \frac{\frac{1}{r}-1}{\left\lfloor \frac{1}{2r}\right\rfloor \left(1-\frac{rK}{\left\lfloor\frac{K}{\left\lfloor \frac{1}{2r}\right\rfloor}\right\rfloor}\right)}\\
& \leq  \frac{\frac{1}{r}-1}{(\frac{1}{2r}-1) \left(1-\frac{rK}{\left\lfloor\frac{K}{\left\lfloor \frac{1}{2r}\right\rfloor}\right\rfloor}\right)}\\
& \leq  \frac{\frac{1}{r}-1}{(\frac{1}{2r}-1) \left(1-\frac{rK}{\left\lfloor 2rK \right\rfloor}\right)}\\
& \overset{(a)}{=} \frac{\frac{1}{r}-1}{\frac{1}{2}(\frac{1}{2r}-1)}\\
& = \frac{4(\frac{1}{r}-1)}{\frac{1}{r}-2}\\
& = 4 + \frac{4}{\frac{1}{r}-2}\\
& < 3+\sqrt{5}, \label{eq:constGap7}
\end{align}
where (a) is because that $rK$ is a positive integer.

Comparing the two bounds in (\ref{eq:constGap6}) and (\ref{eq:constGap7}) completes the proof.
\end{proof}

\section{Coded MapReduce: General Description and Performance Analysis}
In this section, we present our proposed Coded MapReduce scheme on a general MapReduce job, and analyze the corresponding inter-server communication load.

\subsection{Coded MapReduce}
We consider a MapReduce job of evaluating $Q$ keys on an input file with $N$ subfiles using $K$ servers, where each subfile is assigned to $pK$ servers for Map tasks and is actually mapped at $rK$ of them. Before proceeding to detailed descriptions of the Coded MapReduce scheme, we recall that after the Map tasks execution, $\mathcal{M}_k'$ indicates the set of subfiles mapped at Server $k$, and $\mathcal{A}_n'$ indicates the set of servers that have finished mapping Subfile $n$.

\subsection*{Map Tasks Assignment}
We assume that $N$ is sufficiently large and $N = g \begin{pmatrix}K \\ pK \end{pmatrix}$ for some integer $g$ \footnote{Otherwise we can introduce enough empty subfiles to have an integer-valued $g$.}. The master controller partitions the subfiles into $\begin{pmatrix}K \\ pK \end{pmatrix}$ equal-sized subsets, where each subset contains $g$ unique subfiles. We call each of these subsets a batch of subfiles. For each batch of $g$ subfiles, the master controller chooses a distinct subset of $pK$ servers, and assign all $g$ subfiles in the batch to each of these $pK$ servers.

Because each server belongs to $\begin{pmatrix} K-1 \\ pK-1 \end{pmatrix}$ subsets of size $pK$, each server is assigned $g\begin{pmatrix} K-1 \\ pK-1 \end{pmatrix}$ subfiles by the end of Map tasks assignment, i.e., $|\mathcal{M}_k|=g\begin{pmatrix} K-1 \\ pK-1 \end{pmatrix}$ for all $k \in \{1,\ldots,K\}$.

For the motivating example in Section III, $N=12$, ${Q=K=4}$, $pK =2$, thus we have $g=2$. Every $2$ servers are assigned $2$ unique chapters. 

\subsection*{Map Tasks Execution}
After the Map tasks assignment, each server starts to map all of its assigned subfiles simultaneously.  We assume that the times the $K$ servers spend mapping their assigned subfiles are i.i.d. across servers and subfiles. Thus the probability that any subset of $rK$ servers in $\mathcal{A}_n$ finish mapping Subfile $n$ by the end of Map tasks execution is $\dfrac{1}{\begin{pmatrix}pK \\ rK \end{pmatrix}}$, for all ${n \in \{1,\ldots,N\}}$. Because each group of $rK$ servers are assigned $g\begin{pmatrix}K-rK \\pK-rK \end{pmatrix}$ subfiles, by the end of Map tasks execution, the expected number of subfiles mapped at any group of $rK$ servers is 
\begin{equation*}
\dfrac{g\begin{pmatrix}K-rK \\pK-rK \end{pmatrix}}{\begin{pmatrix}pK \\ rK \end{pmatrix}}.
\end{equation*}

By law of large numbers and the fact that ${N = g\begin{pmatrix}K \\pK \end{pmatrix} > g\begin{pmatrix}K-rK \\pK-rK \end{pmatrix}}$, for $N$ large enough the actual number of subfiles mapped at any group of $rK$ servers is 
\begin{equation}\label{eq:r}
\dfrac{g\begin{pmatrix}K-rK \\pK-rK \end{pmatrix}}{\begin{pmatrix}pK \\ rK \end{pmatrix}} + o(N),
\end{equation}
with high probability.

Before proceeding to describe the Reduce tasks, we define an important quantity $\mathcal{V}^{k}_{\mathcal{S}}$ for $k \in \{1,\ldots,K\}$ and ${{\cal S} \subseteq \{1,\ldots,K\}}$ such that $\mathcal{V}^{k}_{\mathcal{S}}$ is a set of information bits required by Server $k$ and known \emph{exclusively} at all servers in ${\cal S}$. More precisely, $\mathcal{V}^{k}_{{\cal S}} \triangleq \left\{v_{qn}:q \in \mathcal{W}_k, k \notin \mathcal{A}_n'= {\cal S} \right\}$.  

\subsection*{Data Shuffling}
The inter-server communication is carried out as follows to exchange the missing values for Reduce functions:
\begin{enumerate}
\item For each subset ${\cal S}$ of $\{1,\ldots,K\}$ such that ${|{\cal S}|\!=\!rK\!+\!1}$ and each $k \in \mathcal{S}$, arbitrarily partition $\mathcal{V}^{k}_{\mathcal{S}\backslash \{k\}}$ into $rK$ disjoint segments each containing $\frac{\left|\mathcal{V}^{k}_{\mathcal{S} \backslash \{k\}}\right|}{rK}$ bits: 
\begin{equation*}
\mathcal{V}^{k}_{\mathcal{S}\backslash \{k\}} = \left\{\mathcal{V}^{k}_{\mathcal{S}\backslash \{k\},i}: i \in \mathcal{S} \backslash \{k\}\right\},
\end{equation*}
where $\mathcal{V}^{k}_{\mathcal{S}\backslash \{k\},i}$ is the segment associated with Server $i$ in $\mathcal{S} \backslash \{k\}$.

\item For each $i \in \mathcal{S}$, Server $i$ zero-pads all associated segments $\left\{\mathcal{V}^{k}_{\mathcal{S}\backslash \{k\},i}: k\in \mathcal{S} \backslash \{i\} \right\}$ to the length of the longest one and sends the following coded segment:
\begin{equation*}
\underset{k \in \mathcal{S} \backslash \{i\}}\oplus \mathcal{V}^{k}_{\mathcal{S}\backslash \{k\},i},
\end{equation*}
where $\oplus$ denotes bitwise XOR on the zero-padded segments. 
\end{enumerate}

After Server $k$, $k \in \mathcal{S}\backslash \{i\}$ receives the coded segment $\underset{k \in \mathcal{S} \backslash \{i\}}\oplus \mathcal{V}^{k}_{\mathcal{S}\backslash \{k\},i}$ from Server $i$, because Server $k$ knows all the information bits in $\mathcal{V}^{k'}_{\mathcal{S}\backslash \{k'\},i}$ for all $k' \in \mathcal{S} \backslash \{i,k\}$, it can cancel them from the coded segment and recover the intended bits in $\mathcal{V}^{k}_{\mathcal{S}\backslash \{k\},i}$.

Consider the the motivating example in Section III. For the subset of servers $\mathcal{S}=\{1,2,3\}$, ${\mathcal{V}^{1}_{\{2,3\}} = \{a_7,a_8\}}$, $\mathcal{V}^{2}_{\{1,3\}} = \{b_3,b_4\}$, and $\mathcal{V}^{3}_{\{1,2\}} = \{c_1,c_2\}$. These sets of values are segmented such that
\begin{itemize}
\item  $\mathcal{V}^{1}_{\{2,3\},2} = \{a_7\}$, $\mathcal{V}^{1}_{\{2,3\},3} = \{a_8\}$,

\item  $\mathcal{V}^{2}_{\{1,3\},1} = \{b_3\}$, $\mathcal{V}^{2}_{\{1,3\},3} = \{b_4\}$,

\item  $\mathcal{V}^{3}_{\{1,2\},1} = \{c_1\}$, $\mathcal{V}^{3}_{\{1,2\},2} = \{c_2\}$.
\end{itemize} 

Then during the date shuffling,
\begin{itemize}
\item Server 1 sends a coded pair $(BC,b_3+c_1)[3,1]$,
\item Server 2 sends a coded pair $(AC,a_7+c_2)[7,2]$,
\item Server 3 sends a coded pair $(AB,a_8+b_4)[8,4]$.
\end{itemize}

%

\begin{remark}
The idea of efficiently creating and exploiting coded multicasting was initially proposed in the context of cache networks in~\cite{maddah2014fundamental, maddah2013decentralized}, and extended in ~\cite{ji2014fundamental, karamchandani2014hierarchical}, where caches pre-fetch  part of the content in a way that coding opportunities can be optimally exploited to reduce network traffic during the content delivery. We notice that interestingly, such opportunities exist in a general MapReduce environment if each subfile is repeatedly mapped at different servers, and we propose Coded MapReduce to utilize them to significantly reduce the shuffling load.  $\hfill \square$
\end{remark}

We summarize the proposed Coded MapReduce scheme in Algorithm 1.
\begin{algorithm}[h]
	\caption{Coded MapReduce}
	\label{alg1}
	\begin{algorithmic}[1]
	\Procedure{Map Tasks Assignment}{}
	\State $g \leftarrow  \left\lceil\dfrac{N}{\begin{pmatrix} K \\ pK \end{pmatrix}} \right\rceil$
	\State $\mathcal{P} \leftarrow \left\{\mathcal{T}: \mathcal{T} \subseteq \{1,\ldots K\}, |\mathcal{T}|=pK\right\}$
	\State Partition the $N$ subfiles into $\left\{\mathcal{U}_{\mathcal{T}}:\mathcal{T} \in \mathcal{P}, |\mathcal{U}_{\mathcal{T}}|=g \right\}$
	\For {$\mathcal{T} \in \mathcal{P}$}
	\State Assign $\mathcal{U}_{\mathcal{T}}$ to all servers in $\mathcal{T}$
	\EndFor
	\EndProcedure
\\
	\Procedure{Data Shuffling}{}
		\For {${\cal S} \subseteq \{1,\ldots,K\}$: $|{\cal S}|=rK+1$}
		\For {$k \in \mathcal{S}$}
		\State $\mathcal{V}^{k}_{{\cal S}\backslash \{k\}} \leftarrow \left\{v_{qn}:q \in \mathcal{W}_k,  \mathcal{A}_n'= {\cal S} \backslash \{k\} \right\}$
		\State Partition $\mathcal{V}^{k}_{\mathcal{S} \backslash \{k\}}$ into $rK$ disjoint segments of equal size: $\left\{\mathcal{V}^{k}_{\mathcal{S}\backslash \{k\},i}: i \in \mathcal{S} \backslash \{k\}\right\}$
		\EndFor
		
		\For {$i \in \mathcal{S}$}
		\State Server $i$ zero-pads all the segments in  $\left\{\mathcal{V}^{k}_{\mathcal{S}\backslash \{k\},i}: k \in \mathcal{S} \backslash \{i\} \right\}$ to have length $\underset{k \in \mathcal{S} \backslash \{i\}}{\max}\left|\mathcal{V}^{k}_{\mathcal{S} \backslash \{k\},i}\right|$ 
		\State Server $i$ sends bitwise XOR of zero-padded segments: $\underset{k \in \mathcal{S} \backslash \{i\}}{\oplus} \mathcal{V}^{k}_{\mathcal{S}\backslash \{k\},i}$
		\EndFor
		
		\EndFor
		\EndProcedure
	\end{algorithmic}
\end{algorithm}

\subsection{Performance of Coded MapReduce}
We first demonstrate that the inter-server communication of Coded MapReduce successfully delivers all required bits for reduction. To start, we pick an arbitrary server, say Server $k$, and consider a required bit to execute its reducers. If this bit is from subfiles that are mapped at Server $k$ (i.e., $\mathcal{M}_k'$), it is readily available for reduction and no communication is needed. However if it is from some Subfile $n$ that is not mapped at Server $k$, i.e., outside $\mathcal{M}_k'$, we recall that \mbox{$\mathcal{A}_n'$ ($k \notin \mathcal{A}_n'$)} denotes the set of servers that have mapped Subfile $n$ and there are $rK$ of them. Now we consider Line 11 of Algorithm 1 for the iteration of $\mathcal{S} = \{k\} \cup \mathcal{A}_n'$. Suppose that this bit is in the segment $\mathcal{V}^{k}_{\mathcal{S}\backslash \{k\},i}$ of some Server $i \in \mathcal{A}_n'$, because for all $k' \in \mathcal{S} \backslash \{i,k\}$, $\{i,k\} \subseteq \mathcal{S} \backslash \{k'\}$ and Server $k$ knows all the bits in $\mathcal{V}^{k'}_{\mathcal{S}\backslash \{k'\},i}$,  
Server $k$ can decode its required bit from the coded segment $\underset{k \in \mathcal{S} \backslash \{i\}}{\oplus} \mathcal{V}^{k}_{\mathcal{S}\backslash \{k\},i}$ sent by Server $i$ by cancelling $\mathcal{V}^{k'}_{\mathcal{S}\backslash \{k'\},i}$ for all $k' \in \mathcal{S} \backslash \{i,k\}$. Similarly, by the end of the shuffling phase, Server $k$ decodes all required bits for its reducers from coded segments transmitted by other servers. The same arguments directly apply to all servers and we have demonstrated that Coded MapReduce successfully provides all values needed for reduction. 

Next, we analytically characterise the inter-server communication load of Coded MapReduce.

\subsection*{Communication Load $L_{\textup{CMR}}(r)$ (Upper Bound on $L^*(r)$ in Theorem 1)}
In each iteration of the inter-server communication in Algorithm 1, by (\ref{eq:r}) we know that after Map tasks execution, every $rK$ servers share the values of all keys in ${\dfrac{g \begin{pmatrix}K-rK \\pK-rK \end{pmatrix}}{\begin{pmatrix}pK \\ rK \end{pmatrix}} + o(N)}$ subfiles, thus for any subset $\mathcal{S}$ of $rK+1$ servers and sufficiently large $N$,
\begin{equation}
\left|\mathcal{V}^{k}_{\mathcal{S}\backslash \{k\}}\right| = \dfrac{Q g\begin{pmatrix}K-rK \\pK-rK \end{pmatrix}}{K \begin{pmatrix} pK \\ rK \end{pmatrix}}F + o(N)F,
\end{equation}
for all $k \in \mathcal{S}$, given that Server $k$ needs $\frac{Q}{K}$ values for its reducers from each subfile.

For each $\mathcal{S} \subseteq \{1,\ldots,K\}$ with $|\mathcal{S}|=rK + 1$, by Lines 17 and 18 of Algorithm 1, Server $i$ sends a coded segment of length 
\begin{equation*}
\dfrac{Q g\begin{pmatrix}K-rK \\pK-rK \end{pmatrix}}{K \begin{pmatrix} pK \\ rK \end{pmatrix} rK}F + o(N)F,
\end{equation*}
for all $i \in \mathcal{S}$. Therefore, the total number of bits sent over the shared link in $\mathcal{S}$ is  
\begin{equation*}
\dfrac{Q g \begin{pmatrix}K-rK \\pK-rK \end{pmatrix}(rK+1)}{K \begin{pmatrix} pK \\ rK \end{pmatrix} rK}F + o(N)F.
\end{equation*}

Because Algorithm 1 has a total of $\begin{pmatrix} K \\ rK+1 \end{pmatrix}$ iterations, the communication load achieved by the proposed Coded MapReduce scheme, normalized by $F$, is
\begin{align*}
L_{\text{CMR}}(r) =& \begin{pmatrix} K  \\rK+1 \end{pmatrix} \!\dfrac{Qg \begin{pmatrix} K-rK \\ pK-rK \end{pmatrix} (rK+1)}{K\begin{pmatrix} pK \\ rK \end{pmatrix} rK} + o(N)\\
=& \dfrac{Q}{K} \times \dfrac{g \,K!}{(rK+1)! \, (K-rK-1)!} \\
&\times \dfrac{(K-rK)! \,(rK+1)}{(pK-rK)! \, (K-pK)!}\\
&\times \dfrac{(rK)! \, (pK-rK)!}{(pK)!\, rK} +o(N)\\
\overset{(a)}{=}& \dfrac{QN}{K}\left(\dfrac{1}{r}-1\right)+o(N),
\end{align*}
where (a) is because that $N = g\begin{pmatrix}K \\ pK \end{pmatrix}$.


\section{Lower Bounds on $L^*(r)$}
In this section, two lower bounds on the inter-server communication load $L^*(r)$ are derived by applying the cut-set bounds on the compound extensions of multiple valid reducer distributions (i.e., which servers reduce which keys). Before proceeding to the formal proofs, we re-emphasize that as highlighted in Remark~1, the communication loads of all valid reducer distributions are identical. As a first step, we partition the indices of the keys $\{1,\ldots,Q\}$ evenly into $K$ groups $\mathcal{G}_1,\ldots,\mathcal{G}_K$ such that $\mathcal{G}_{i} = \left\{\frac{(i-1)Q}{K}+1,\ldots, \frac{iQ}{K}\right\}$, for all $i \in \{1,\ldots,K\}$. Given a reducer distribution $\mathcal{D} = \left(\mathcal{W}_1,\ldots,\mathcal{W}_K\right)$ that specifies which keys are reduced at which servers, we denote the messages sent by Server $k$ during the data shuffling as $X_k^{\mathcal{D}}$, for all $k \in \{1,\ldots,K\}$ with $R_{k}^{\mathcal{D}}F$ total number of bits.

\subsection*{First Bound}
We consider the following $K$ valid reducer distributions on the $K$ servers:
\begin{equation}\label{eq:reducer_dis}
\begin{aligned}
\mathcal{D}_1 &= \left(\mathcal{G}_1,\mathcal{G}_2,\ldots,\mathcal{G}_K\right),\\
\mathcal{D}_2 &= \left(\mathcal{G}_2,\mathcal{G}_3,\ldots,\mathcal{G}_1\right),\\
& \vdots\\
\mathcal{D}_K &= \left(\mathcal{G}_K,\mathcal{G}_1,\ldots,\mathcal{G}_{K-1}\right).
\end{aligned}
\end{equation}

For a particular reducer distribution $\mathcal{D}_i$, $i \in \{1,\ldots,K\}$, we denote its $k$th entry (i.e., the set of indices of the keys reduced at Server $k$) as $\mathcal{D}_{i,k}$. 

Now we consider the compound setting of all $K$ reducer distributions at Server $k$, $k \in \{1,\ldots,K\}$. Server $k$ needs to evaluate all keys in $\left\{\mathcal{D}_{i,k}: i \in \{1,\ldots,K\} \right\}$, which is equal to the set of all keys $\{1,\ldots,Q\}$ by the construction of the reducer distributions in~(\ref{eq:reducer_dis}). During the shuffling phase for a particular reducer distribution $\mathcal{D}_i$, Server $k$ receives $\sum \limits_{k' \neq k} R_{k'}^{\mathcal{D}_i}F$ information bits from other servers, and overall $\sum \limits_{i=1}^K \sum \limits_{k' \neq k} R_{k'}^{\mathcal{D}_i}F$ bits of information across all $K$ reducer distributions. Together with its mapping outcomes $\{v_{qn}: q \in \{1,\ldots,Q\}, n \in \mathcal{M}_k'\}$, Server $k$ should be able to evaluate all $Q$ keys after receiving all messages from all other servers in all reducer distributions. Thus we are essentially considering a cut at Server $k$ separating $\left(\{v_{qn}: q \in \{1,\ldots,Q\}, n \in \mathcal{M}_k'\}, \left\{X_{k'}^{\mathcal{D}_1},\ldots,X_{k'}^{\mathcal{D}_K}: k' \neq k \right\}\right)$ and $\left\{v_{qn}:q \in \{1,\ldots,Q\}, n \in \{1,\ldots,N\}\right\}$, and we have
\begin{equation}
Q|\mathcal{M}_k'|F + \sum \limits_{i=1}^K \sum \limits_{k' \neq k} R_{k'}^{\mathcal{D}_i}F \geq QNF.
\end{equation}

Summing up the cut-set bounds of all $K$ servers, we have
\begin{align}
\sum \limits_{k=1}^K Q |\mathcal{M}_{k}'|F +  \sum \limits_{k=1}^K \sum \limits_{i=1}^K \sum \limits_{k' \neq k} R_{k'}^{\mathcal{D}_i}F &\geq KQNF, \nonumber\\
\Rightarrow  rKQN + \sum \limits_{i=1}^K (K-1)\sum \limits_{k=1}^K R_{k}^{\mathcal{D}_i} &\overset{(b)}{\geq}  KQN, \\
\Rightarrow rKQN + K(K-1)T &\overset{(c)}{\geq}  KQN, \label{eq:perServer}
\end{align}
where (b) is due to the Map tasks execution that each subfile is completely mapped at $rK$ servers, and (c) is because that by Remark~1 the communication load of the data shuffling scheme is independent of the reducer distribution.

Taking expectations of both sides of (\ref{eq:perServer}) over the Map tasks executions $\{\mathcal{M}_k'\}_{k=1}^K$, we obtain the following lower bound of the minimum inter-server communication load: 
\begin{equation}
L^*(r) \geq QN\dfrac{1-r}{K-1}.
\end{equation}

\subsection*{Second Bound}
Let $s \in \{1,\ldots,K\}$ and we focus on the keys evaluated at the first $s$ servers. We consider a set of $\left\lfloor \frac{K}{s} \right\rfloor$ valid reducer distributions such that the indices of the keys reduced by the first $s$ servers are respectively specified as:
\begin{equation}\label{eq:partial}
\begin{aligned}
\mathcal{D}_{1,(1,\ldots,s)} &= \left(\mathcal{G}_1,\ldots,\mathcal{G}_s\right),\\
\mathcal{D}_{2,(1,\ldots,s)} &= \left(\mathcal{G}_{s+1},\ldots,\mathcal{G}_{2s}\right),\\
& \vdots\\
\mathcal{D}_{\left\lfloor \frac{K}{s} \right\rfloor,(1,\ldots,s)} &= \left(\mathcal{G}_{\left(\left\lfloor \frac{K}{s} \right\rfloor-1\right)s+1},\ldots,\mathcal{G}_{\left\lfloor \frac{K}{s} \right\rfloor s}\right),
\end{aligned}
\end{equation}
where $\mathcal{D}_{i,(1,\ldots,s)}$, $i \in \{1,\ldots,\left\lfloor \frac{K}{s} \right\rfloor\}$ denotes indices of the keys evaluated by the first $s$ servers in $i$th reducer distribution. 

Now for the compound setting of the $\left\lfloor \frac{K}{s} \right\rfloor$ reducer distributions in (\ref{eq:partial}), Server $k$, $k \in \{1,\ldots,s\}$, needs to reduce the keys in $\left\{\mathcal{G}_{(i-1)s+k}: i\in \{1,\ldots, \left\lfloor \frac{K}{s} \right\rfloor\}\right\}$, given the transmitted messages for all reducer distributions $\left\{X_{k}^{\mathcal{D}_i}:k \in \{1,\ldots,K\}, i\in \{1,\ldots, \left\lfloor \frac{K}{s} \right\rfloor\}\right\}$ and its mapping outcomes $\{v_{qn}:q \in \{1,\ldots,Q\},n \in \mathcal{M}_k'\}$.

We consider the cut containing all $s$ servers that separates all values needed for reduction and the transmitted messages in all reducer distributions in addition to their local mapping outcomes, i.e.,  $\left\{v_{qn}:q \in \left\{1,\ldots,\frac{Q}{K} \left\lfloor \frac{K}{s} \right\rfloor s \right\}, n \in \{1,\ldots,N\}\right\}$ and $\left(\left\{v_{qn}:q \in \{1,\ldots,Q\}, n \in \mathcal{M}_k'\right\}_{k=1}^{s}, \left\{X_{k}^{\mathcal{D}_i}:k \in \{1,\ldots,K\}, i\in \{1,\ldots, \left\lfloor \frac{K}{s} \right\rfloor\}\right\}\right)$, we have
\begin{align}
Q\sum \limits_{k=1}^s |\mathcal{M}_k'|F + \sum \limits_{i=1}^{\left\lfloor \frac{K}{s} \right\rfloor} \sum \limits_{k=1}^K R_k^{\mathcal{D}_i}F &\geq \frac{Q}{K} \left\lfloor \frac{K}{s} \right\rfloor s NF, \nonumber \\
\Rightarrow Q\sum \limits_{k=1}^s |\mathcal{M}_k'| + \left\lfloor \frac{K}{s} \right\rfloor T &\geq \left\lfloor \frac{K}{s} \right \rfloor \dfrac{sQN}{K}. \label{eq:bound2}
\end{align}
Taking expectations of both sides of (\ref{eq:bound2}) over the Map tasks executions $\{\mathcal{M}_k'\}_{k=1}^K$ and using the fact that $\mathbb{E}\{|\mathcal{M}_k'|\}=rN$ for all $k \in \{1,\ldots,K\}$, we have
\begin{align}
srQN + \left\lfloor \frac{K}{s} \right\rfloor L(r) &\geq \left\lfloor \frac{K}{s} \right \rfloor \dfrac{sQN}{K}, \nonumber \\
\Rightarrow L(r) &\geq  sQN \left(\dfrac{1}{K}-\dfrac{r}{\left\lfloor \frac{K}{s} \right \rfloor}\right). \label{eq:firstS}
\end{align}

Because (\ref{eq:firstS}) holds for all $s \in \{1,\ldots,K\}$, we obtain another lower bound of $L^*(r)$:
\begin{equation}
L^*(r) \geq  \underset{s \in \{1,\ldots, K\}}{\max}  sQN \left(\dfrac{1}{K}-\dfrac{r}{\left\lfloor \frac{K}{s} \right \rfloor}\right). \label{eq:secLower}
\end{equation}

For the example in Section~III where $Q=4$, $N=12$, $K=4$ and $r = \frac{1}{2}$, the first lower bound is tighter than the second, yielding that $L^*(\frac{1}{2}) \geq 8$. 

\section{Map Processing Time vs. Communication Load of Coded MapReduce}
As Coded MapReduce slashes the inter-server communication load through repetitive Map operations across servers, processing more Map tasks incurs extra delay to the job execution. In this section, we analytically and numerically investigate this effect and the tradeoff between the processing time of the Map tasks and the inter-server communication load of Coded MapReduce. 

\subsection{Map Processing Time of Coded MapReduce}
We first recall that after the Map tasks assignment, each server starts to process the assigned $pN$ subfiles until each of the $N$ subfiles has been mapped at $rK$ distinct servers. We assume that all subfiles are of the same size and are simultaneously available for processing. We employ the processor-sharing (PS) model in~\cite{Kleinrock,Coffman1970} as our service discipline where the processing rate of each server, denoted as $\mu$, is evenly distributed among all assigned $pN$ Map tasks. We assume that the $K$ servers start their respective Map tasks at the same time and the Map processing times of the subfiles within and across servers are i.i.d. exponential random variables with rate $\frac{\mu}{pN}$ \footnote{For simplicity, we do not consider the rate change due to the completions of Map tasks and we assume that all subfiles are being mapped at the constant rate $\frac{\mu}{pN}$ throughout the Map tasks execution.}.  Suppose that the Map tasks are carried out error-free, then the Map processing time of a single subfile say Subfile $n$, $S_n$ is the $rK$th-order statistic of $pK$ i.i.d. exponential random variables with rate $\frac{\mu}{pN}$, and the probability density function of $S_n$ by~\cite{papoulis2002probability} is
\begin{align}
f_{S_n}(s) &= pK f(s) \begin{pmatrix}pK-1 \\ rK-1\end{pmatrix} F(s)^{rK-1}(1-F(s))^{pK-rK}\nonumber \\
&= pK \frac{\mu}{pN} e^{-\frac{\mu}{pN} s} \begin{pmatrix}pK-1 \\ rK-1\end{pmatrix}(1-e^{-\frac{\mu}{pN} s})^{rK-1}e^{-\frac{\mu}{pN} s(pK-rK)} \nonumber\\
&= \frac{K}{N}\mu \begin{pmatrix}pK-1 \\ rK-1\end{pmatrix}(1-e^{-\frac{\mu}{pN} s})^{rK-1}e^{-\frac{\mu}{pN}(pK-rK+1)s}, \label{eq:pdf}
\end{align}
where $f(s)=\frac{\mu}{pN} e^{-\frac{\mu}{pN} s}$ and $F(s) = 1-e^{-\frac{\mu}{pN} s}$ are the PDF and CDF of an exponential random variable with rate $\frac{\mu}{pN}$ respectively.

We can further derive the CDF of the processing time of Subfile $n$ (i.e., $S_n$) as
\begin{align}
F_{S_n}(s) &= \int_{0}^{s} f_{S_{n}}(t) dt \nonumber \\
&= \sum \limits_{j=0}^{rK-1}pK \begin{pmatrix}pK-1 \\ rK-1\end{pmatrix} \begin{pmatrix}rK-1 \\ j\end{pmatrix}(-1)^{rK-1-j}\frac{1-e^{-\frac{\mu}{pN}(pK-j)s}}{pK-j}.\label{eq:cdf}
\end{align}

The average processing time to map Subfile $n$, $n\in \{1,\ldots,N\}$, at $rK$ servers out of the $pK$ servers it is assigned to, i.e., the expected value of $S_n$, is (see e.g.~\cite{arnold1992first})
\begin{equation}
\mathbb{E}\{S_n\} = \frac{pN}{\mu}\sum \limits_{j=1}^{rK} \frac{1}{pK+1-j}.
\end{equation}

Because the Map tasks execution continues until each of the $N$ subfiles is mapped at $rK$ servers, the overall Map processing time of $N$ subfiles, $S$ is
\begin{equation*}
S = \max\{S_1,\ldots,S_N\},
\end{equation*}
where $S_1,\ldots,S_N$ are i.i.d. random variables with the probability density function in~(\ref{eq:pdf}) and the cumulative distribution function in~(\ref{eq:cdf}).

The distribution function of the overall processing time $S$ can be calculated as $F_S(s) = \prod \limits_{n=1}^N \textup{Pr}(S_n \leq s) = \left(F_{S_{n}}(s)\right)^N$, and thus the average overall processing time for the Map tasks is
\begin{align*}
\mathbb{E}\{S\} &= \int_{0}^{\infty} (1-F_{S}(s)) ds\\
&= \int_{0}^{\infty} \left(1-\left(F_{S_{n}}(s)\right)^N \right) ds,
\end{align*}
where $F_{S_n}(s)$ is calculated in (\ref{eq:cdf}).

\subsection{Numerical Evaluations}
In Fig.~\ref{fig:tradeoff_indi} and Fig.~\ref{fig:tradeoff}, we have respectively plotted the average processing time to map one subfile and to map all subfiles and its corresponding shuffling load, using Coded MapReduce. As expected, when increasing the number of repetitive mapping ($rK$), the Map processing time becomes longer while less communication is required for shuffling. Coded MapReduce provides a flexible framework for system designers to balance between the time spent in the Map phase and the shuffling phase of MapReduce jobs, based on the underlying infrastructures (e.g., server processing speed and network speed), to optimize the run-time performance. 

\begin{figure}[htbp]
   \centering
   \includegraphics[width=0.5\textwidth]{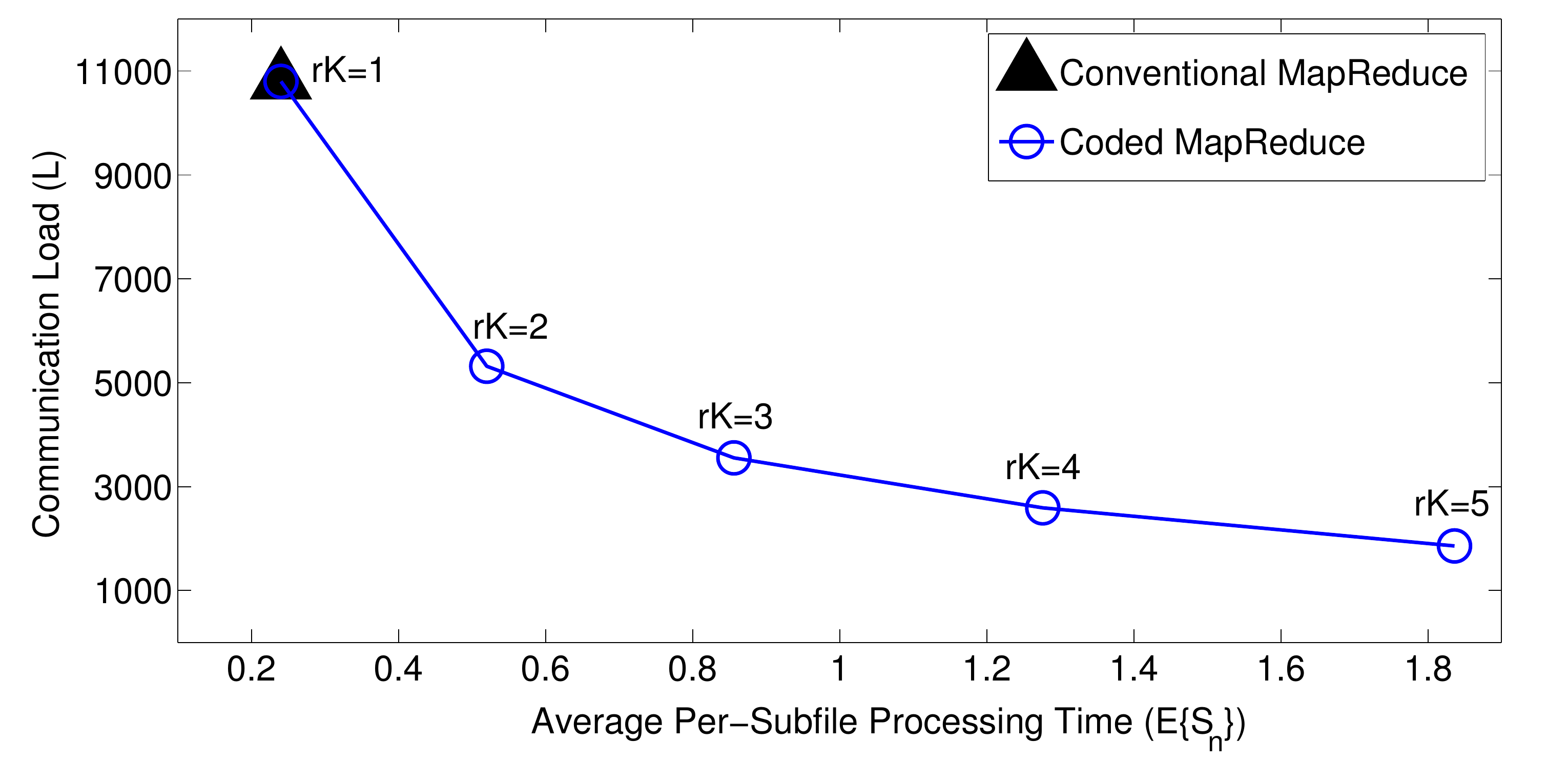}
   \caption{The performance of the Coded MapReduce for a job with $N=1200$ subfiles, $Q=10$ keys, $K=10$ servers and  $pK=7$ repetitive assignments of each subfile. The Map processing rate at each server $\mu = 500$. This figure demonstrates the tradeoff between the processing time to map one subfile at $rK$ servers and the communication load in the following shuffling phase.}
   \label{fig:tradeoff_indi}
   \vspace{-4 mm}
\end{figure}

\begin{figure}[htbp]
   \centering
   \includegraphics[width=0.48\textwidth]{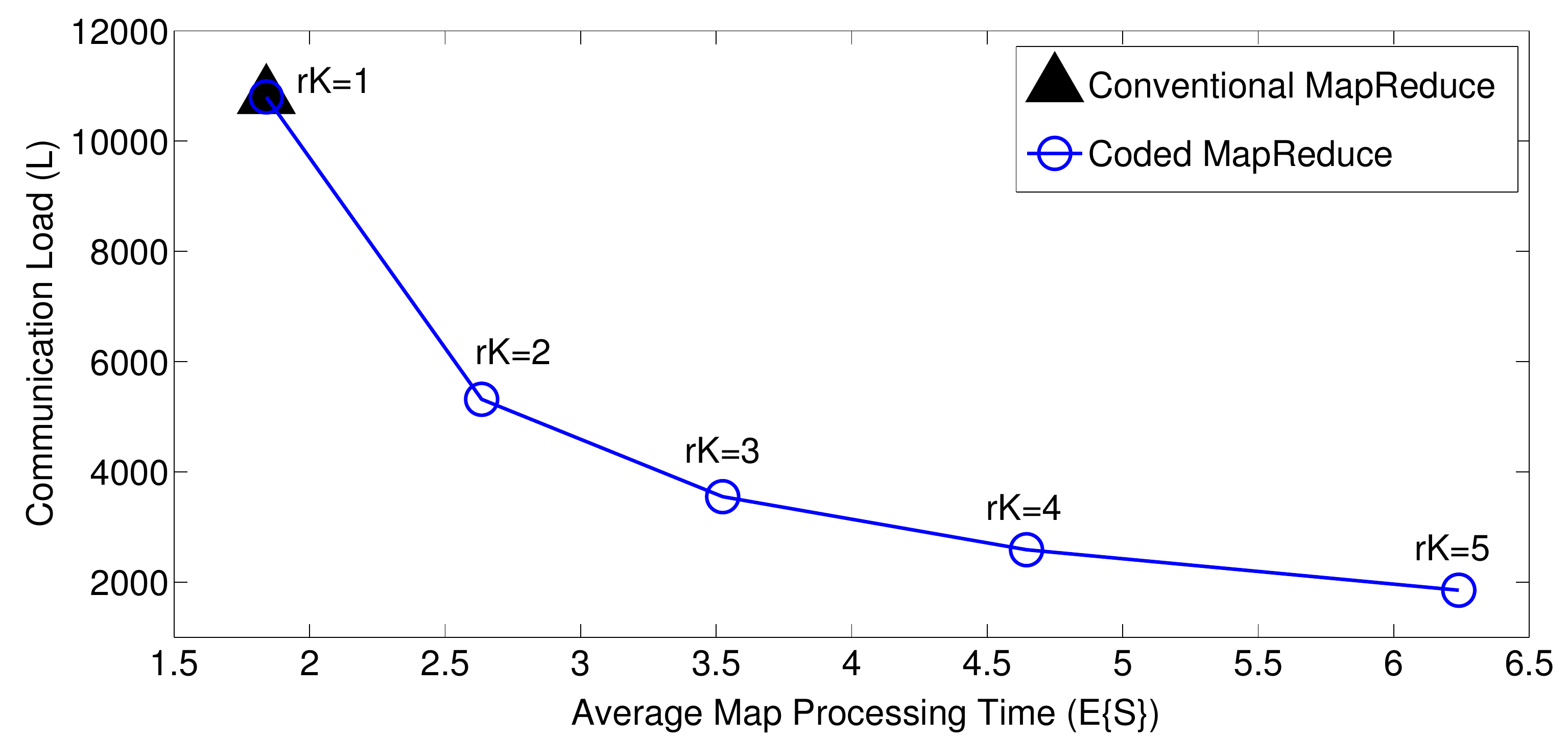}
   \caption{The performance of the Coded MapReduce for a job with $N=1200$ subfiles, $Q=10$ keys, $K=10$ servers and  $pK=7$ repetitive assignments of each subfile. The Map processing rate at each server $\mu = 500$. This figure demonstrates the tradeoff between the overall Map processing time and the communication load in the following shuffling phase.}
   \label{fig:tradeoff}
     \vspace{-2 mm}
\end{figure}

\section{Conclusion}
In this paper, we propose Coded MapReduce, a joint framework that aims to improve the performance of data shuffling in a MapReduce environment. We demonstrate through analytical and numerical analysis that, by carefully assigning repetitive Map tasks onto different servers and smartly coding the transmitted message bits across keys and data blocks, Coded MapReduce can slash the inter-server communication load by a factor that grows linearly with number of servers in the system. We also demonstrate the optimality of Coded MapReduce by showing that it achieves the minimum communication load within a constant multiplicative factor. Further, we explore the inherent tradeoff of Coded MapReduce  between the computation load of the Map tasks and the communication load of the data shuffling, shedding light on optimizing Coded MapReduce to minimize the overall job completion time.  An interesting future direction is to consider a software implementation of Coded MapReduce in Hadoop clusters to demonstrate the coding gain for reducing the inter-server communication load of the shuffling phase in such clusters.


\bibliographystyle{IEEEtran}
\bibliography{ref}

\end{document}